\documentclass[letterpaper,12pt]{article}
\usepackage{tabularx} 
\usepackage{amsmath}  
\usepackage{graphicx} 
\usepackage[margin=1in,letterpaper]{geometry} 
\usepackage{cite} 
\usepackage[final]{hyperref} 
\hypersetup{
	colorlinks=true,       
	linkcolor=blue,        
	citecolor=blue,        
	filecolor=magenta,     
	urlcolor=blue         
}
\usepackage{blindtext}
\usepackage[utf8]{inputenc} 
\usepackage[T1]{fontenc}    
\usepackage{url}            
\usepackage{booktabs}       
\usepackage{amsfonts}       
\usepackage{nicefrac}       
\usepackage{microtype}      
\usepackage{lipsum}
\usepackage[braket]{qcircuit}
\usepackage{braket}
\usepackage{subfigure}
\usepackage{amsthm}
\usepackage{mathtools}
\usepackage[ruled,linesnumbered]{algorithm2e}
\usepackage{authblk}
\usepackage{float}
\usepackage{enumerate}
\usepackage{tablefootnote}
\usepackage{comment}
\usepackage{array}
\usepackage{bbm}
\usepackage{mathrsfs}
\usepackage{amssymb}
\usepackage{cleveref}

\numberwithin{equation}{section}

\newtheorem{theorem}{Theorem}[section]
\newtheorem{lemma}[theorem]{Lemma}
\newtheorem{proposition}[theorem]{Proposition}
\newtheorem{definition}[theorem]{Definition}
\newtheorem{fact}[theorem]{Fact}
\newtheorem{problem}[theorem]{Problem}

\makeatletter

\newcommand{\Rmnum}[1]{\expandafter\@slowromancap\romannumeral #1@}
\makeatother

\begin{document}

\title{Tolerant Quantum Junta Testing}

\author{Zhaoyang Chen\thanks{chenchy76@mail2.sysu.edu.cn}}

\author{Lvzhou Li\thanks{lilvzh@mail.sysu.edu.cn}}

\author{Jingquan Luo\thanks{luojq25@mail2.sysu.edu.cn}}
\affil{Institute of Quantum Computing and Software, \\School of Computer Science and Engineering, \\Sun Yat-sen University, Guangzhou 510006, China}
\date{}
\maketitle

\begin{abstract}
Junta testing for Boolean functions has sparked a long line of work over recent decades in theoretical computer science, and recently has also been studied for unitary operators in quantum computing. Tolerant junta testing is more general and challenging than the standard version. While optimal tolerant junta testers have been obtained for  Boolean functions, there has been  no knowledge about tolerant junta testers for unitary operators, which was thus left as an open problem in [Chen,  Nadimpalli, and Yuen, SODA2023]. In this paper, we settle this problem by presenting the first algorithm to decide  whether a unitary is $\epsilon_1$-close to some quantum $k$-junta or is $\epsilon_2$-far from any quantum $k$-junta, where an $n$-qubit unitary $U$ is called a quantum $k$-junta if it only non-trivially acts on just $k$ of the $n$ qubits.
More specifically, we present a tolerant tester with $\epsilon_1 = \frac{\sqrt{\rho}}{8} \epsilon$, $\epsilon_2 = \epsilon$, and  $\rho \in (0,1)$, and the query complexity is $O\left(\frac{k \log k}{\epsilon^2 \rho (1-\rho)^k}\right)$, which demonstrates a trade-off between the amount of tolerance and the query complexity. Note that our algorithm is non-adaptive which is preferred over its adaptive counterparts,  due to its simpler as well as highly parallelizable nature. At the same time, our algorithm does not need access to $U^\dagger$, whereas this is usually required in the literature.

\end{abstract}

\section{Introduction}
Characterization of the dynamical behavior of quantum systems is a significant task in physics. In general, quantum process tomography \cite{chuang1997prescription,poyatos1997complete} can be used for extracting  information about the quantum operation, yet it consumes large amounts of resources. When provided with oracle access to an $n$-qubit unitary operator, the process necessitates $\Omega(4^n)$\cite{gutoski2014process} queries to the oracle, for deriving a classical representation of the unitary. However, in many situations, we merely want to know whether the unitary operator satisfies a certain property or is far from having the property, rather than acquiring a complete understanding of the unitary. This can potentially lead to a considerable decrease in query complexity. Similar problems have been studied in both classical\cite{goldreich1998property,blum1990self} and quantum\cite{montanaro2013survey} settings, within the framework called \textit{property testing} outlined below.

\paragraph{Property Testing} Let $\mathcal{X}$ be a set of objects and define a distance metric $\mathrm{dist}:\mathcal{X} \times \mathcal{X} \rightarrow [0,1]$. A subset $\mathcal{P} \subseteq \mathcal{X}$ is referred to as a property. We say that an object $x$ is $\epsilon$-close to $\mathcal{P}$ if there exists a element $y \in \mathcal{P}$ such that $\mathrm{dist}(x,y) \le \epsilon$; conversely, $x$ is $\epsilon$-far from $\mathcal{P}$ if $\mathrm{dist}(x,y) > \epsilon$ for every $y \in \mathcal{P}$. A property tester for $\mathcal{P}$ receives either an object $x\in \mathcal{P}$ or $x$ that is $\epsilon$-far from $\mathcal{P}$. In the former case, The tester will accept the input with high probability, while in the latter case, it will reject with high probability.

There has been a lot of research into efficient testers for various properties and proofs demonstrating that specific properties are not efficiently testable within the realm of classical computing\cite{blum1990self,goldreich1998property,fischer2004art,ron2008property,10.5555/1980617}. In contrast,  there is a relative scarcity of studies focused on quantum algorithms for property testing. Nevertheless, the field of quantum property testing has received increasing attention in recent years, including quantum testing of classical properties\cite{bernstein1993quantum,ambainis2007quantum,grover1996fast}, classical testing of quantum properties\cite{Krauthgamer2003PropertyTO,van2000self} and quantum testing of quantum properties\cite{wang2011property,harrow2013testing}, as highlighted in \cite{montanaro2013survey}.

A standard property testing algorithm is only guaranteed to accept the given object, if the object \textit{exactly} satisfies the property, which is somewhat restrictive. What if one wishes to accept objects
that are close to the desired property? Such instances tend to appear due to noise and other imperfections. To address this
question, Parnas et al.\cite{parnas2006tolerant} introduced a intuitive generalization of the standard property testing model, called \textit{tolerant property testing}. In this model, the tester is expected to accept any object if it is \textit{close} to the desired property while rejecting those that are significantly distant from having the property. \Cref{fig:tolerant_diagram} provides an illustration of the tolerant property testing problem. 

\begin{definition} [Tolerant property testing]
    For constants $0 < \epsilon_1 < \epsilon_2 < 1$ and a property $\mathcal{P}$, a $(\epsilon_1,\epsilon_2)$ tolerant tester for $\mathcal{P}$ is an algorithm that given a object $x$,
    \begin{itemize}
        \item accepts with probability at least $2/3$ if $x$ is $\epsilon_1$-close to $\mathcal{P}$;
        \item rejects with probability at least $2/3$ if $x$ is $\epsilon_2$-far from $\mathcal{P}$.
    \end{itemize}
\end{definition}

    \begin{figure}[htb]
    \centering
    \includegraphics[width=0.5\linewidth]{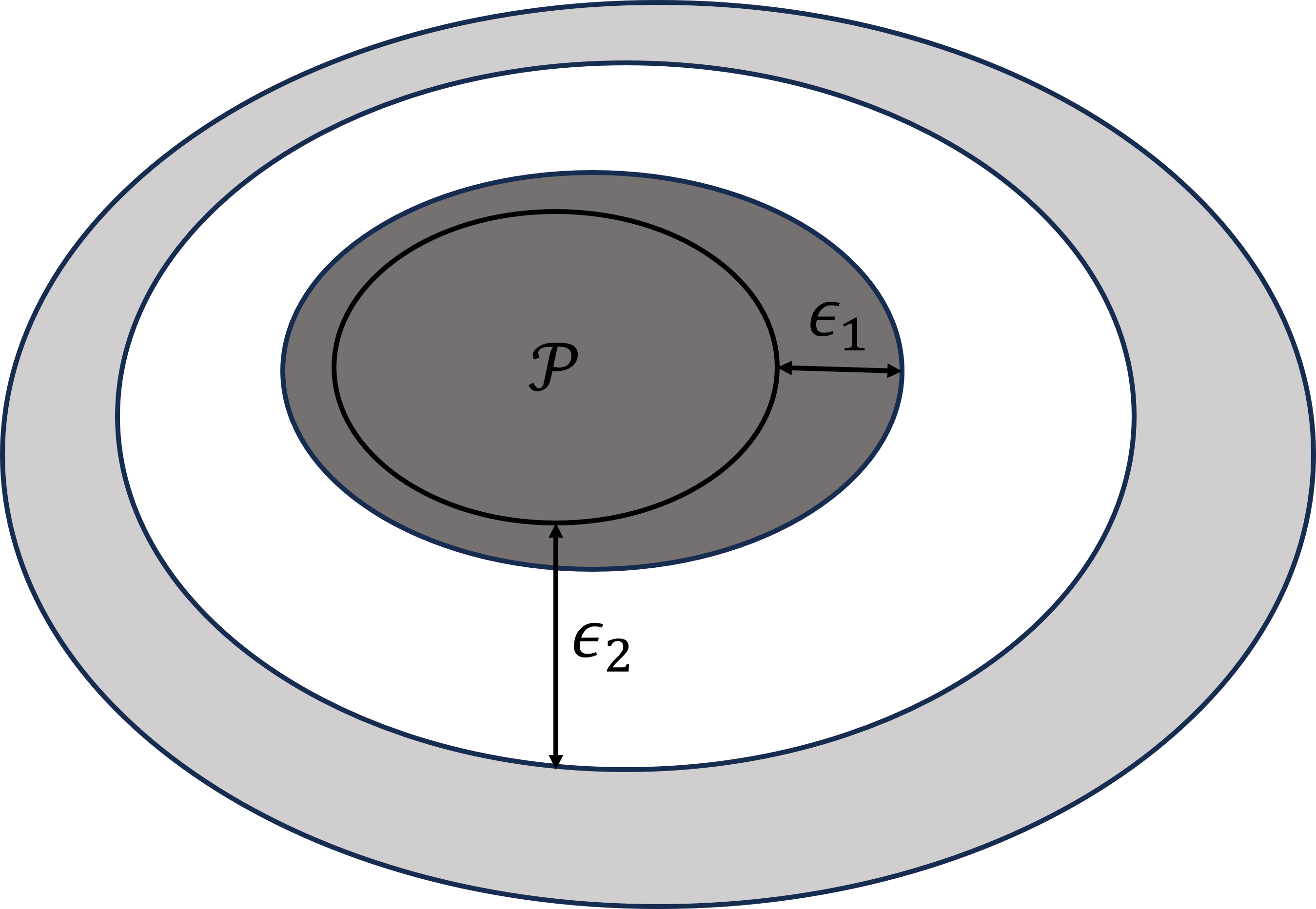}
    \caption{A schematic diagram of tolerant property testing. The outside oval represents the set of objects $\mathcal{X}$ and the innermost oval represents the desired property. The role of the tolerant tester is to distinguish between the dark grey and the light grey areas.}
    \label{fig:tolerant_diagram}
    \end{figure}

The problem of tolerant property testing could be much more challenging than the version of standard property testing. For instance, classically, the nearly tight upper bound $O(k\log k + k/\epsilon)$ for testing $k$-juntas was established by Blais\cite{blais2009testing} in 2009, whereas the tight bound $2^{\widetilde{O}(\sqrt{k\log (1/\epsilon)})}$ for tolerant $k$-junta testing  was provided by Nadimpalli and Patel\cite{nadimpalli2024optimal} until very recently. There has also been some research on junta testing via quantum algorithms (refer to \Cref{section:related work} for details). However, as previously mentioned by \cite{chen2023testing}, there has been a lack of investigation into tolerant junta
testing in quantum computing.

\subsection{Our result}
This paper focuses on the problem of quantum junta testing.  An $n$-qubit unitary $U$ is called a \textit{ quantum $k$-junta} if it only non-trivially acts on just $k$ of the $n$ qubits (for a formal statement, refer to \Cref{definition:quantum k-junta}). 
Given two unitaries $U$ and $V$:
  \begin{itemize}
        \item $U$ is said to be  \textit{$\epsilon$-close} to $V$, if $\mathrm{dist}(U,V) \le\epsilon$,
\item  $U$ is \textit{$\epsilon$-far} from $V$, if $\mathrm{dist}(U, V) > \epsilon$. 
\end{itemize}
The distance metric $\mathrm{dist}(.,.)$ is given in Definition  \ref{definition:distance metric}. In this paper, we investigate the issue of tolerant quantum junta testing given in the following.

\begin{problem} [Tolerant quantum junta testing] \label{problem:toleran testing}
    For parameters $0< \epsilon_1 < \epsilon_2 < 1$ and $k \ge 1$, when provided with oracle access to an $n$-qubit unitary $U$, the task is to determine if $U$ is $\epsilon_1$-close to some quantum $k$-junta or if $U$ is $\epsilon_2$-far from every quantum $k$-junta.
\end{problem}
 Our result for the above problem is stated below. 
\begin{theorem} \label{theorem:main result}
    For parameters $\epsilon \in (0,1)$, $\rho \in (0,1)$ and $k \ge 1$, there exists an algorithm that given oracle access to a unitary $U$  on $n$ qubits,  satisfies the following guarantees:
    \begin{itemize}
        \item If $U$ is $\frac{\sqrt{\rho}}{8} \epsilon$-close to a quantum k-junta,  the algorithm outputs accept with probability at least $2/3$;
        \item If $U$ is $\epsilon$-far from every quantum $k$-junta, the algorithm outputs reject with probability at least $2/3$.
    \end{itemize}
    The query complexity of this algorithm is $O\left(\frac{k \log k}{\epsilon^2 \rho (1-\rho)^k}\right)$.
\end{theorem}

There are some remarks:

(i) To our best knowledge, this is the first result on tolerant  testing  $k$-juntas in quantum computing. It has been noted  in \cite{chen2023testing} ``there has been no work on tolerant property testing—
for both Boolean functions as well as unitary matrices—via quantum algorithms''.  Actually, the problem of tolerant property testing is  much more challenging than the version of standard property testing, which could be reflected by:  classically, the nearly tight upper bound for testing $k$-juntas was established by Blais\cite{blais2009testing} in 2009, whereas the tight bound  for tolerant $k$-junta testing  was provided by Nadimpalli and Patel\cite{nadimpalli2024optimal} until very recently; quantumly, there is blank about telorant testing via quantum algorithms for both Boolean functions and quantum operations except for our result as shown in  \Cref{table:related work}.

(ii) Our algorithm is non-adaptive and requires no access to $U^\dagger$. There is a lot of work in classical computing dedicated to designing non-adaptive algorithms for $k$-junta testing (e.g., \cite{de2019junta,nadimpalli2024optimal}), since non-adaptive algorithms are preferred over their adaptive counterparts,
due to their simpler as well as highly parallelizable nature.  The existing quantum testers for both Boolean $k$-juntas \cite{ambainis2007quantum} and  quantum $k$-juntas \cite{chen2023testing} are adaptive and require  access to $U^\dagger$ as well as $U$, since they use quantum amplitude amplification. As noted in \cite{montanaro2013survey},  the assumption of access to $U^\dagger$  may not be reasonable in an adversarial scenario where we only assume access to $U$ as a black box. Luckily, our algorithm  is non-adaptive and requires no access to $U^\dagger$.

(iii) Although our algorithm does not fully resolve the tolerant testing problem, it achieves a trade-off between the query complexity and the amount of tolerance. When $\rho = \Omega(1)$, the algorithm is capable of distinguishing between unitaries that are $O(\epsilon)$-close to some quantum $k$-junta and unitaries that are $\epsilon$-far from every  quantum $k$-junta, with query complexity $2^{O(k)}/\epsilon^2$. When $\rho = O(1/k)$, the algorithm can differentiate between unitaries that are $O(\epsilon/\sqrt{k})$-close to some quantum $k$-junta and unitaries that are $\epsilon$-far from  every quantum $k$-junta, with query complexity $\widetilde{O}(k^2/\epsilon^2)$.
We hope this will spark more discussion on quantum tolerant testing.

\subsection{Technique Overview }
 
Our tolerant quantum junta tester  makes use of the notion of {\it influence of qubits on a unitary}\footnote{We use $\mathbf{Inf}_U[S]$ to denote the influence of the qubits from $S$ on unitary $U$, and the detailed  definition is given in \Cref{definition:influence of qubits}},  first introduced by Montanaro and Osborne \cite{montanaro2010quantum}  in the context of Hermitian unitary matrices, and  playing a crucial role in the standard version of quantum junta testing \cite{chen2023testing}. More importantly, we appeal to the concept of {\it quantum $k$-part junta} defined in this paper which is a quantum analog to $k$-part junta in  \cite{blais2019tolerant}.

Our work is mainly inspired by \cite{blais2019tolerant} which is for  torelant testing of Boolean $k$-juntas, but  we have to  extend some crucial steps  carefully and skillfully to the quantum realm.
The main idea behind our algorithm is as follows. 
\begin{enumerate}
\item [(i)]\textbf{Reduction from testing quantum $k$-juntas to testing quantum $k$-part juntas.} We first randomly partition the $n$ qubits into $l=O(k^2)$ parts denoted as $\mathcal{I}=\{I_1,\dots,I_l\}$. Then we show the following:  if  $U$ is close to some quantum $k$-junta, then it is close to some quantum $k$-part junta with respect to $\mathcal{I}$, and on the contrary, if $U$ is far from every quantum $k$-junta, then it is far from every quantum $k$-part junta with respect to $\mathcal{I}$ (see \Cref{proposition:reduction to part juntas}). Therefore, the problem is now to determine whether $U$ is close to some quantum $k$-part junta, which  in turn is equivalent to determining whether there  exist $k$ parts out of $\mathcal{I}$ such that the  qubits  from  their complement  has  a small influence (see  \Cref{definition:partition quantum juntas}). To determine if there exist such  $k$ parts, one can exhaustively consider all subsets $S\subseteq [l]$ with  $|S|=l-k$ and estimate the influence $\mathbf{Inf}_U[\phi_{\mathcal{I}}(S)]$ of the set $\phi_{\mathcal{I}}(S) = \bigcup_{i\in S} I_i$. 
\item [(ii)]\textbf{Reducing the query complexity by $\rho$-biased subset.} If one estimate $\mathbf{Inf}_U[\phi_{\mathcal{I}}(S)]$  for each $S$ by a naive strategy, then it will lead to a high query complexity to $U$. The query complexity can be reduced if one query to $U$ can be repeatedly used to estimate the influence  $\mathbf{Inf}_U[\phi_{\mathcal{I}}(S)]$ for different choices of $S$.
For this reason, we introduce the concept of \textit{$\rho$-biased subset} as follows.
Consider a set $S\subseteq [l]$ and a parameter $\rho \in (0,1)$. A random $\rho$-biased subset $S'\sim_\rho S$ is generated by adding each index from $S$ into $S'$ with probability $\rho$. We can prove that, for any set $S \subseteq [l]$, the expected value of $\mathbf{Inf}_U[\phi_{\mathcal{I}}(S')]$ satisfies \begin{align}\frac{\rho}{3} \mathbf{Inf}_U[\phi_\mathcal{I}(S)] \le \mathbf{E}_{S'\sim_\rho S}\left[\mathbf{Inf}_U[\phi_{\mathcal{I}}(S')]\right] \le \mathbf{Inf}_U[\phi_\mathcal{I}(S)],\label{eq1}\end{align}
which is given in \Cref{lemma:rho influence and influence}.
\item[(iii)] \textbf{Estimating the expected value of  the $\rho$-biased subset's influence.} We propose an algorithm to estimate $\mathbf{E}_{S'\sim_\rho S}\left[\mathbf{Inf}_U[\phi_{\mathcal{I}}(S')]\right]$ for all $S\subseteq [l]$ with size $l-k$, which thus implies an estimation for $\mathbf{Inf}_U[\phi_{\mathcal{I}}(S)]$ according to \Cref{eq1}. (See \Cref{algorithm:rho influence estimator}.)
\end{enumerate}

Put all the above together, we obtain a tolerant tester for quantum $k$-juntas. 

\subsection{Related Work} \label{section:related work}

A Boolean function $f:\{0,1\}^n \rightarrow \{0,1\}$ is a $k$-junta if its value only depends on at most $k$ variables. Such functions are often studied both in computational learning theory\cite{mossel2003learning,valiant2015finding} and machine learning\cite{luz}, particularly in scenarios where the majority of features are irrelevant for the concept that we want to learn. For instance, in the field of computational biology, some genetic attribute is only determined by a few genes on the long DNA sequence. The related work and our contribution are presented in \Cref{table:related work}.
\begin{table}[!ht]
    \centering
    \Huge
    \renewcommand\arraystretch{2}
    \resizebox{\columnwidth}{!}{
    \begin{tabular}{c|c|c|c}
    \hline
    &  $f:\{0,1\}^n \rightarrow \{0,1\}$ &   \textbf{Unitary $U$ on $n$ qubits} &   \textbf{Channel $\Phi$ on $n$ qubits}   \\ \hline
    \textbf{ Classical Standard Testing} & $O(k\log k + k/\epsilon)$\cite{blais2009testing} &  /  & / \\
     & $\Omega(k\log k)$\cite{sauglam2018near} & / & /\\ \hline
     \textbf{ Quantum Standard Testing} & $\widetilde{O}(\sqrt{k/\epsilon})$\cite{ambainis2016efficient} & $\widetilde{O}(\sqrt{k}/\epsilon)$\cite{chen2023testing} & $\widetilde{O}(k/\epsilon^2)$\cite{bao2023testing} \\
     & $\widetilde{\Omega}(\sqrt{k})$\cite{bun2018polynomial} & $\widetilde{\Omega}(\sqrt{k})$\cite{chen2023testing} & $\widetilde{\Omega}(\sqrt{k})$\cite{bao2023testing} \\ \hline
       \textbf{Classical Tolerant Testing }&  $2^{\widetilde{O}(\sqrt{k\log(1/\epsilon)})}$\cite{nadimpalli2024optimal} & / & /\\
      & $2^{\widetilde{\Omega}(\sqrt{k\log (1/\epsilon)})}$\cite{nadimpalli2024optimal} & / & /\\ \hline
      \textbf{ Quantum Tolerant Testing} & / & $O\left(\frac{k \log k}{\epsilon^2 \rho (1-\rho)^k}\right)$(\Cref{theorem:main result}) & /\\
      & / & / & /\\
    \hline
    \end{tabular}
    }
    \caption{Our result and prior work on testing Boolean and quantum $k$-juntas.}
    \label{table:related work}
\end{table}

\paragraph{Classical Standard Testing of Boolean Juntas.}
Classically, in the standard testing model, the first $k$-juntas tester was raised by Fischer et al.\cite{fischer2004testing} with query complexity $\widetilde{O}(k^2)$. More recently, the best known upper bound for testing $k$-juntas $O(k\log k + k/ \epsilon)$ was provided by Blais\cite{blais2009testing}, which is tight for constant $\epsilon$\cite{sauglam2018near}. There has also been some research on testing $k$-juntas in the \textit{distribution free}\footnote{In the distribution free setting, the distribution on inputs is not assumed to be uniform.} setting. The best known upper bound for this problem is $\widetilde{O}(k/\epsilon)$ due to Bshouty\cite{bshouty:LIPIcs.CCC.2019.2} and Zhang\cite{Zhang2019NearOptimalAF}, which is optimal up to logarithmic factors.

\paragraph{Classical Tolerant Testing of Boolean Juntas.}
Diakonikolas et al.\cite{diakonikolas2007testing} first considered the problem of tolerant junta testing and found that the standard tester from \cite{fischer2004testing} can also provide a $(\mathrm{poly}(\epsilon/k),\epsilon)$ tolerant tester. Chakraborty et al.\cite{chakraborty2012junto} subsequently observed that the analysis in \cite{blais2009testing} actually implies a $(\epsilon/C,\epsilon)$ tolerant tester (for some constant $C > 1$) with $\exp(k/\epsilon)$ queries.
Blais et al.\cite{blais2019tolerant} gave a $(\rho \epsilon/16,\epsilon)$ tolerant tester for some $\rho \in (0,1)$ using $O\left(\frac{k \log k}{\epsilon \rho (1-\rho)^k}\right)$ queries.
Notice that there is a multiplicative "gap" between $\epsilon_1$ and $\epsilon_2$ in the above testers such that $\epsilon_1$ and $\epsilon_2$ can not be arbitrarily close. To address this, De, Mossel, and Neeman\cite{de2019junta} obtained a non-adaptive $(\gamma, \gamma +\epsilon)$ tolerant tester, with $2^k \cdot \mathrm{poly}(k,\frac{1}{\epsilon})$ queries. Subsequently Iyer, Tal and Whitmeyer\cite{iyer_et_al:LIPIcs.CCC.2021.24} introduced an adaptive $2^{\widetilde{O}(\sqrt{k/\epsilon})}$-query $(\gamma, \gamma + \epsilon)$ tolerant tester. Recently Nadimpalli and Patel\cite{nadimpalli2024optimal} gave a non-adaptive $(\gamma, \gamma + \epsilon)$ tolerant tester that makes $2^{\widetilde{O}(\sqrt{k\log(1/\epsilon)})}$ queries, with a matching lower bound. It is worth mentioning that the tolerant tester by Blais et al.\cite{blais2019tolerant} could perform much better in some case, i.e., when $\rho = O(1/k)$ this yields a $(O(\epsilon/k),\epsilon)$ tester using $\mathrm{poly}(k)/\epsilon$ queries.

\paragraph{Quantum Standard Testing of Boolean and Quantum Juntas.}
There has also been some research on testing Boolean juntas via quantum algorithms. Atıcı and Servedio\cite{atici2007quantum} introduced the first quantum algorithm based on Fourier sampling using $O(k/\epsilon)$ queries. Subsequently, Ambainis et al.\cite{ambainis2016efficient} suggested an improved upper bound $\widetilde{O}(\sqrt{k/\epsilon})$, which was proved to be optimal up to a polylogarithmic factor by Bun, Kothari, and Thaler\cite{bun2018polynomial}. In the distribution free setting, Belovs\cite{belovs2019quantum} gave a $O(k/\epsilon)$ quantum tester, matching the best upper bound of classical algorithms\cite{bshouty:LIPIcs.CCC.2019.2,Zhang2019NearOptimalAF}.

It is reasonable to explore the testing of quantum operations in quantum computing besides Boolean functions. Building on a generalization of the tester introduced by Atıcı and Servedio\cite{atici2007quantum}, Wang\cite{wang2011property} developed  a tester for testing quantum $k$-juntas using $\widetilde{O}(k/\epsilon^2)$ queries. In a recent contribution, Chen, Nadimpalli and Yuen\cite{chen2023testing} established a nearly tight upper bound $\widetilde{O}(\sqrt{k}/\epsilon)$ for quantum junta testing. Additionally, Bao and Yao\cite{bao2023testing} presented the first  junta  testing algorithm with $\widetilde{O}(k/\epsilon^2)$ queries for quantum channels.

\subsection{Discussion}
There are several  problems worthy of further consideration:
\begin{enumerate}
    \item Our work has not completely resolved the problem of tolerant quantum junta testing in \Cref{problem:toleran testing}. There exists an inevitable multiplicative gap between $\epsilon_1$ and $\epsilon_2$, i.e., $\epsilon_2 > 8\epsilon_1$ for all choices of $\rho$. We expect a tolerant quantum junta tester capable of handling values of $\epsilon_1$ and $\epsilon_2$ that are arbitrarily close.
    \item Kearns and Ron\cite{kearns1998testing} introduced another generalization of the standard property testing model called "parameterized". The parameterized tester is to determine whether $x \in \mathcal{P}$ or $x$ is $\epsilon$-far from $\mathcal{P}' \supseteq \mathcal{P}$. Note that if $\mathcal{P}'$ is a proper superset of $\mathcal{P}$, the testing problem becomes easier, and normally the parameterized tester will require a smaller query or sample complexity than the standard tester. There have been no quantum algorithms for the parameterized testing problem.
    \item We expect to establish the quantum query lower bound for the problem of tolerant quantum junta testing.
    \item  We are curious to determine whether quantum algorithms offer any advantage for the problem of tolerant Boolean junta testing, compared to the classical lower bound $2^{\widetilde{\Omega}(\sqrt{k\log (1/\epsilon)})}$ obtained by Nadimpalli and Patel\cite{nadimpalli2024optimal}. There seems to be blank on this topic.
\end{enumerate}

\subsection{Organization}

The organization of this paper is as follows. we introduce some essential notations and provide relevant background information in \Cref{section:preliminaries}. Then in \Cref{section:reduction} we demonstrate that the problem of testing quantum juntas can be reduced to testing quantum partition juntas. Lastly, we present the tolerant tester and the proof of \Cref{theorem:main result} in \Cref{section:solution}.

\section{Preliminaries} \label{section:preliminaries}

In this section, we present the essential notations and definitions. For a positive integer $n$, we use $[n]$ to denote the set $\{1,2,\dots,n\}$. For any subset $T\subseteq [n]$, we denote its complement as $\overline{T} \coloneqq [n]\backslash T$.  We use $|S|$ to denote the cardinality of  set $S$.

\subsection{Intersecting Families}
A family $\mathcal{F}$ of subsets of $[n]$ is $t$-\textit{intersecting} if for any two sets $S,T \in \mathcal{F}$, the intersection satisfies the condition $|S \cap T| \ge t$. For $0 < p < 1$, we define the $p$-\textit{biased} measure of the family $\mathcal{F}$ as
$$\mu_p (\mathcal{F}) = \Pr_{S}[S \in \mathcal{F}],$$
where $S$ is a random subset of $[n]$, generated by independently assigning each element $i\in [n]$ to $S$ with probability $p$. Our objective is to figure out the maximum $p$-biased measure of a $t$-intersecting family, given a fixed $p$. Friedgut\cite{friedgut2008measure}, Dinur and Safra\cite{dinur2005hardness} established a upper bound as detailed in \Cref{lemma:biased measure bound}.

\begin{lemma}[\cite{dinur2005hardness,friedgut2008measure}] \label{lemma:biased measure bound}
    Consider a $t$-intersecting family $\mathcal{F}$ of subsets of $[n]$ for some $t \ge 1$. For any $p< \frac{1}{t+1}$, the $p$-biased measure of $\mathcal{F}$ is bounded by $\mu_p(\mathcal{F}) \le p^t$.
\end{lemma}

\subsection{Probability}
We recall the following Chernoff bound and the union bound used in this paper.
\begin{fact}[The Chernoff bound] \label{fact:chernoff bound}
     Suppose $X_1,\dots,X_N$ are independent random variables taking values in $\{0,1\}$, let $X \coloneqq \sum_{i=1}^N X_i$, $\mu = \mathbf{E}[X] $, for any $\delta > 0$, we have
     $$\Pr \left[X > (1+\delta)\mu \right] \le e^{-\frac{\delta^2 \mu}{2+\delta}},\ \Pr \left[X < (1-\delta)\mu \right] \le e^{-\frac{\delta^2 \mu}{2}},$$
     and
     $$\Pr \left[|X -\mu| > \delta \mu \right] \le 2e^{-\frac{\delta^2 \mu}{2+\delta}}.$$
\end{fact}

\begin{fact} [The union bound] \label{fact:union bound}
    For a collection of events $A_1,\dots,A_N$, it holds that
    $$\Pr\left[ \bigcup_{i=1}^N A_i \right] \le \sum_{i=1}^{N}\Pr\left[A_i \right].$$
\end{fact}

\subsection{Unitary Operators}
 Throughout this paper, we use $\mathcal{U}_n$ to denote the set of all unitary operators on $n$ qubits.

\begin{definition}[Quantum $k$-junta] \label{definition:quantum k-junta}
    Given $T\subseteq [n]$, we say that a unitary $U\in \mathcal{U}_n$ is a quantum junta on $T$, if it can be expressed in the form
    $$U = V_T\otimes I_{\overline{T}},$$    
    where $V_T $ acts on $T$. Furthermore, $U$ is called a quantum $k$-junta, if  $|T|=k$. We denote the class of quantum $k$-juntas as $\mathcal{V}_k$.
\end{definition}

For any two operators $A,B$ in a Hilbert space, the Hilbert-Schmidt inner product is defined as $\left \langle A,B \right \rangle = \mathrm{Tr}(A^\dagger B)$. We introduce a metric for the purpose of evaluating the distance between two operators.

\begin{definition} \label{definition:distance metric}
    Given two operators $A,B$ acting on an $N$-dimensional Hilbert space, we define
    $$\mathrm{dist}(A,B) \coloneqq \min_{\theta \in [0,2\pi)} \frac{1}{\sqrt{2N}} \|e^{i\theta}A-B\|,$$
    where $\|A\| =\sqrt{\mathrm{Tr}(A^\dagger A)}$ is the Frobenius norm. 
\end{definition}

The normalization factor $1/\sqrt{2N}$ is to ensure  $0\le \mathrm{dist}(A,B) \le 1$ for unitary operators. Observed that $\mathrm{dist}(A,B) = 0$ if and only if $A = e^{i\theta}B$ for some $\theta \in [0,2\pi)$. In addition, for unitary operators, there is a relation between $\mathrm{dist}(U,V)$ and $\langle U, V\rangle:$
\begin{equation*}\label{equation:relation between d and product}
    \mathrm{dist}(U,V) = \sqrt{1-\frac{1}{N}|\langle U, V\rangle}|.
\end{equation*}

\subsection{Influence of Qubits on Unitaries}
Recall the Pauli operators
$$\sigma_0 = I = \begin{pmatrix}
    1 & 0 \\
    0 & 1
\end{pmatrix},\ \sigma_1 = X = \begin{pmatrix}
    0 & 1 \\
    1 & 0 
\end{pmatrix},\ \sigma_2 = Y = \begin{pmatrix}
    0 & -i \\
    i & 0
\end{pmatrix},\ \sigma_3 = Z = \begin{pmatrix}
    1 & 0 \\
    0 & -1
\end{pmatrix}.$$
For any vector $x = (x_1,x_2,\dots,x_n) \in \{0,1,2,3\}^n \cong \mathbb{Z}_4^n$, we define $\sigma_x \coloneqq \otimes_{i=1}^n \sigma_{x_i}$. Then the set $\left \{\frac{1}{\sqrt{2^n}} \sigma_x \right\}_{x\in \mathbb{Z}_4^n}$ constitutes an orthonormal basis for the Hilbert space consisting of all operators on $n$ qubits. Consequently, any operator $A$ can be expressed in the form 
$$A = \sum_{x\in \mathbb{Z}_4^n} \widehat{A}(x) \sigma_x,$$
where the term $\widehat{A}(x) = \frac{1}{2^n} \left \langle A,\sigma_x \right \rangle$ represents the \textit{Pauli\ coefficient}, and the set $\{\widehat{A}(x)\}_x$ is referred to as the \textit{Pauli\ spectrum} of $A$. It is straightforward to confirm that
$$\frac{1}{2^n} \|A\|^2 = \sum_{x\in \mathbb{Z}_4^n}|\widehat{A}(x)|^2.$$
Additionally,  for any unitary operator $U \in \mathcal{U}_n$, it holds that $\sum_{x\in \mathbb{Z}_4^n}|\widehat{U}(x)|^2=1$.

We can now move forward to introduce the definition of influence of qubits on unitaries, first introduced by Montanaro and Osborne \cite{montanaro2010quantum}  in the context of Hermitian unitary matrices.
\begin{definition} [Influence of qubits] \label{definition:influence of qubits}
    For a unitary operator $U \in \mathcal{U}_n$ and a set $S \subseteq [n]$, the influence of $S$ on $U$ is defined as
    $$\mathbf{Inf}_U[S] =  \sum_{x:\mathrm{supp}(x)\cap S \ne \emptyset} |\widehat{U}(x)|^2,$$
    where $\mathrm{supp}(x)\coloneqq \{i\in [n]:x_i \ne 0\}$.
\end{definition}

The two evident properties of influence in \Cref{lemma:properties of influence} will be utilized consistently throughout our discussion.

\begin{lemma} \label{lemma:properties of influence}
Consider a unitary operator $U\in \mathcal{U}_n$ and two sets $S,T\subseteq [n]$. The following properties hold:
\begin{itemize}
    \item Monotonicity: If $S \subseteq T$, then $\mathbf{Inf}_U[S] \le \mathbf{Inf}_U[T]$;
    \item Subadditivity: $\mathbf{Inf}_U[S\cup T] \le \mathbf{Inf}_U[S] + \mathbf{Inf}_U[T]$.
\end{itemize}
\end{lemma}

Here we present two crucial lemmas that describe the relations between distance and influence.

\begin{lemma} \label{lemma:close to k-juntas}
    Given a unitary $U\in \mathcal{U}_n$, if $U$ is $\epsilon$-close to a quantum junta on $T$, then
    $$\mathbf{Inf}_U [\overline{T}] \le 2\epsilon^2.$$
\end{lemma}

\begin{proof}
    Let unitary $V = V_T \otimes I_{\overline{T}}$ such that $\mathrm{dist}(U,V) \le \epsilon$.  Define the set $\mathbb{Z}_4^T = \{x\in \mathbb{Z}_4^n : \mathrm{supp}(x) \subseteq T\}$. We have $\widehat{V}(x) = 0$ if $x \notin \mathbb{Z}_4^T$. Therefore,
    \begin{equation}
        \begin{aligned}
            \mathrm{dist}(U,V) = \sqrt{1 - \frac{1}{2^n}|\braket{U, V}| } = \sqrt{1 - \left|\sum_{x \in \mathbb{Z}_4^T} \widehat{U}^\dagger(x) \widehat{V}(x)\right| } \leq \epsilon.
        \end{aligned}
    \end{equation}
    According to Cauchy-Schwarz inequality, we have 
    \begin{equation}
        \begin{aligned}
            \sum_{x \in \mathbb{Z}_4^T}  \left| \widehat{U}(x) \right|^2 \geq \left|\sum_{x \in \mathbb{Z}_4^T} \widehat{U}^\dagger(x) \widehat{V}(x)\right|^2 \geq (1 - \epsilon^2)^2 > 1 - 2\epsilon^2.
        \end{aligned}
    \end{equation}
    Therefore, 
    \begin{equation}
        \begin{aligned}
            \mathbf{Inf}_U[\overline{T}] = 1- \sum_{x\in \mathbb{Z}_4^T} |\widehat{U}(x)|^2 < 2\epsilon^2.
        \end{aligned}
    \end{equation}
\end{proof}

The following result was implied in \cite{wang2011property}, although the notion of influence was not given there.
\begin{lemma}[\cite{wang2011property}] \label{lemma:far from k-juntas}
    Given a unitary $U \in \mathcal{U}_n$, if $U$ is $\epsilon$-far from every quantum $k$-junta, it follows that for all $T\subseteq [n]$ with size $|T|\le k$, we have  
    $$\mathbf{Inf}_U[\overline{T}] > \frac{\epsilon^2}{2}.$$
\end{lemma}

\subsection{The Choi-Jamiolkowski Isomorphism} 
The Choi-Jamiolkowski Isomorphism\cite{jamiolkowski1972linear,choi1975completely} shows that there exists a duality between quantum channels and quantum states. In particular, we can associate each $n$-qubit unitary $U$  with a specific Choi-Jamiolkowski state (abbreviated as CJ state):
$$|v(U)\rangle \coloneqq (U\otimes I) \left(\frac{1}{\sqrt{2^n}} \sum_{0\le i < 2^n} |i\rangle |i\rangle \right)= \frac{1}{\sqrt{2^n}} \sum_{0\le i,j < 2^n} U[i,j] |i\rangle |j\rangle.$$

The CJ state $|v(U)\rangle$ can be constructed by generating $n$ EPR pairs $(l,\tilde{l}) \in [n] \times \{n+1,\dots,2n\}$, followed by applying the unitary $U$ on the first half of each EPR pairs. We regard qubits labeled $\{1,\dots,n\}$ as the ones on which $U$ operates, while the qubits labeled $\{n+1,\dots,2n\}$ are influenced by the identity operator $I$.

\section{Reducing to Quantum \texorpdfstring{$k$}{}-part Junta Testing} \label{section:reduction}
In this section, we state how to reduce the problem of testing quantum $k$-juntas to the problem of testing quantum $k$-part juntas. A significant benefit of this approach is that the parameter $n$ is no longer involved. We begin by the definition of quantum $k$-part juntas, analogous to the definition of Boolean $k$-part juntas in \cite{blais2019tolerant}, and establish two relations between quantum $k$-juntas and quantum $k$-part juntas. On the basis of these two relations, we successfully demonstrate the reduction in \Cref{proposition:reduction to part juntas}.

\begin{definition}[Partition quantum juntas] \label{definition:partition quantum juntas}
    Let $\mathcal{I}=\{I_1,\dots,I_l\}$ represent a random partition of $[n]$ into $l \ge k$ parts, where $k \ge 1$. A unitary $U \in \mathcal{U}_n$ is a quantum $k$-part junta with respect to $\mathcal{I}$ if $U$ is a quantum junta on  a union of at most $k$ parts of $\mathcal{I}$. Furthermore, 
    \begin{itemize}
        \item We say that $U$ is $\epsilon$-close to some quantum $k$-part junta with respect to $\mathcal{I}$ if there exists a set $S \subseteq [l]$ of size $l-k$ such that $\mathbf{Inf}_U[\phi_\mathcal{I}(S)] \le \epsilon$;
        \item Conversely, we say that $U$ is $\epsilon$-far from every quantum $k$-part junta with respect to $\mathcal{I}$ if for every set $S\subseteq [l]$ of size $l-k$, it holds that $\mathbf{Inf}_U[\phi_\mathcal{I}(S)] > \epsilon$,
    \end{itemize}
    where $\phi_\mathcal{I}(S) = \bigcup_{i\in S} I_i$. 
\end{definition}

In the following lemma, we prove that the unitaries that are close to some quantum $k$-junta, are also close to some quantum $k$-part junta with respect to any $l$-partition with $l \ge k$.

\begin{lemma} \label{lemma:approximate be k-part}
    Let $\mathcal{I}$ be any partition of $[n]$ with $l \ge k$ parts. If a unitary $U$ is $\epsilon$-close to some quantum $k$-junta, it follows that $U$ is $2\epsilon^2$-close to some quantum $k$-part junta with respect to $\mathcal{I}$.    
\end{lemma}

\begin{proof}
    Let $V \in \mathcal{V}_k$ such that $\mathrm{dist}(U,V) = \mathrm{dist}(U,\mathcal{V}_k) \le \epsilon$. Let $I_{i_1},\dots, I_{i_r}$ be the $r \le k$ parts of $\mathcal{I}$ containing all the relevant qubits $T$ on $V$. For any set $S\subseteq [l]$ with $|S|=k$ such that $\{i_1,\dots,i_r\} \subseteq S$, we have $T \subseteq \phi_\mathcal{I}(S)$. Then $\overline{S} \subseteq [l]$ with $|\overline{S}| = l-k$ and $\phi_\mathcal{I}(\overline{S}) \subseteq \overline{T}$. By \Cref{lemma:close to k-juntas} and the monotonicity of influence, we can conclude that
    $$\mathbf{Inf}_U[\phi_\mathcal{I}(\overline{S})] \le \mathbf{Inf}_U[\overline{T}] \le 2\epsilon^2.$$
\end{proof}

In \Cref{lemma:violate be k-part} we describe a way to create a random partition, and prove that with probability at least $5/6$, the unitaries that are far from every quantum $k$-junta are also far from every quantum $k$-part junta with respect to the partition.

\begin{lemma} \label{lemma:violate be k-part}
    Let $\mathcal{I}$ represent a random partition of the set $[n]$, consisting of $l \overset{\mathrm{def}}{=} 24k^2$ distinct parts, which is generated by uniformly and independently including each element $i\in [n]$ in one of these parts. If we consider a unitary $U$ that is $\epsilon$-far from every quantum $k$-junta, then with probability at least $5/6$ over the choice of the partition $\mathcal{I}$, it follows that $U$ is $\frac{\epsilon^2}{8}$-far from every quantum $k$-part junta with respect to $\mathcal{I}$.
\end{lemma}

\begin{proof}
    For $\tau > 0$, we denote $\mathcal{F}_\tau = \{T \subseteq [n]: \mathbf{Inf}_U[\overline{T}] \le \tau\}$ as the family of subsets whose complements exhibit an influence not exceeding $\tau$. For any two sets $S,T\in \mathcal{F}_{\frac{\epsilon^2}{4}}$, the subadditivity property of influence leads to the conclusion that 
    $$\mathbf{Inf}_U[\overline{S\cap T}] = \mathbf{Inf}_U[\overline{S} \cup \overline{T}] \le \mathbf{Inf}_U[\overline{S}] + \mathbf{Inf}_U[\overline{T}] \le 2 \cdot \frac{\epsilon^2}{4} = \frac{\epsilon^2}{2}.$$
    Given that $U$ is $\epsilon$-far from quantum $k$-juntas, it follows from \Cref{lemma:far from k-juntas} that for any set $J\subseteq [n]$ with size $|J|\le k$, we have the inequality $\mathbf{Inf}_U[\overline{J}] > \frac{\epsilon^2}{2}$, which implies that $|S\cap T| > k$. Moreover, every pair of sets in $\mathcal{F}_{\frac{\epsilon^2}{4}}$ agrees with this argument, such that $\mathcal{F}_{\frac{\epsilon^2}{4}}$ constitutes a $(k+1)$-intersecting family.

    Now we consider two separate cases: either $\mathcal{F}_{\frac{\epsilon^2}{4}}$ contains a set with size less than $2k$, or $\mathcal{F}_{\frac{\epsilon^2}{4}}$ contains only sets with size at least $2k$. In the first case, let us consider a set $S \in \mathcal{F}_{\frac{\epsilon^2}{4}}$ such that $|S| < 2k$. The probability that $S$ is entirely separated by the partition $\mathcal{I}$ is at least
    \begin{equation*}
        \begin{aligned}
            \frac{l-1}{l} \cdot \frac{l-2}{l} \cdots \frac{l-(2k-1)}{l} > \left(1-\frac{2k}{l}\right)^{2k} > 1-2k\cdot \left(\frac{2k}{l}\right) = \frac{5}{6},
        \end{aligned}
    \end{equation*}
    where the second inequality is derived from the binomial theorem. For every set $T\in \mathcal{F}_{\frac{\epsilon^2}{4}}$, it holds that $|S \cap T| \ge k+1$. Consequently, when $S$ is entirely separated by the partition $\mathcal{I}$, no set $T$ in $\mathcal{F}_{\frac{\epsilon^2}{4}}$ can be covered by a union of $k$ parts in $\mathcal{I}$, i.e., for every set $S \subseteq [l]$ of size $k$, we have $\mathbf{Inf}_U[\phi_\mathcal{I}(\overline{S})] > \frac{\epsilon^2}{4} > \frac{\epsilon^2}{8}$. Therefore, with probability at least $\frac{5}{6}$, the unitary $U$ is $\frac{\epsilon^2}{8}$-from every quantum $k$-part junta with respect to $\mathcal{I}$.

    In the second case, where the family $\mathcal{F}_{\frac{\epsilon^2}{4}}$ contains only sets with size at least $2k$, we claim that $\mathcal{F}_{\frac{\epsilon^2}{8}}$ is a $2k$-intersecting family; otherwise, we could find sets $S,T \in \mathcal{F}_{\frac{\epsilon^2}{8}}$ such that $|S\cap T| < 2k$ and $\mathbf{Inf}_U[\overline{S\cap T}] \le \mathbf{Inf}_U[\overline{S}] + \mathbf{Inf}_U[\overline{T}] \le \frac{\epsilon^2}{4}$, which would contradict our assumption about the family $\mathcal{F}_{\frac{\epsilon^2}{4}}$. Let $S\subseteq [n]$ be a union of $k$ parts in $\mathcal{I}$. The set $S$ can be viewed as a random subset, formed by independently including each element of $[n]$ with probability $p = \frac{k}{l} = \frac{1}{24k} < \frac{1}{2k + 1}$. According to \Cref{lemma:biased measure bound}, we find that
    $$\Pr_{\mathcal{I}}\left[\mathbf{Inf}_U[\overline{S}] \le \frac{\epsilon^2}{8}\right] = \Pr\left[S \in \mathcal{F}_{\frac{\epsilon^2}{8}}\right] = \mu_{p}\left(\mathcal{F}_{\frac{\epsilon^2}{8}}\right) \le \left(\frac{k}{l}\right)^{2k}.$$
    
    By applying the union bound over all possible $S$, we can conclude that $U$ is $\frac{\epsilon^2}{8}$-far from every quantum $k$-part junta with respect to $\mathcal{I}$ with probability at least
    $$1-\binom{l}{k}\left(\frac{k}{l}\right)^{2k} \ge 1-  \left(\frac{el}{k}\right)^{k}\left(\frac{k}{l}\right)^{2k} = 1- \left(\frac{ek}{l}\right)^k > \frac{5}{6}.$$
\end{proof}

Consequently, with the two above lemmas, if we want to distinguish between unitaries that are $\epsilon_1$-close to some quantum $k$-junta and unitaries that are $\epsilon_2$-far from every quantum $k$-junta, it suffices to differentiate between unitaries that are $2\epsilon_1^2$-close to some quantum $k$-part junta and unitaries that are $\frac{\epsilon_2^2}{8}$-far from every quantum $k$-part junta with respect to a random partition, with a high probability. The formal statement is given in the following proposition.

\begin{proposition} \label{proposition:reduction to part juntas}
    Assume that $\mathcal{A}$ is an algorithm that provided with oracle access to unitary $U$, parameters $k \ge 1$, $\epsilon \in (0,1)$ and a partition $\mathcal{I}=\{I_1,\dots,I_l\}$ of $[n]$ into $l\ge k$ parts, satisfies the following statements, with $q(k,\epsilon,l)$ queries to $U$ and a function $r : (0,1) \rightarrow (0,1)$.
    \begin{itemize}
        \item In the case where $U$ is $\epsilon'$-close to some quantum $k$-part junta with respect to $\mathcal{I}$ and $\epsilon' \le r(\epsilon)$, the algorithm $\mathcal{A}$ returns accept with probability at least $5/6$;
        \item Conversely, if $U$ is $\epsilon$-far from every quantum $k$-part junta with respect to $\mathcal{I}$, then $\mathcal{A}$ returns reject with probability at least $5/6$.
    \end{itemize}
    It can be demonstrated that there exists an algorithm $\mathcal{A}'$, when provided with oracle access to $U$ and parameters $k \ge 1$, $\epsilon \in (0,1)$ and a function $r : (0,1) \rightarrow (0,1)$, satisfies the following statements.
    \begin{itemize}
        \item If $\mathrm{dist}(U,\mathcal{V}_k)\le \sqrt{\epsilon'/2}$ and $\epsilon' \le r(\epsilon)$, then $\mathcal{A}'$ outputs accept with probability at least $2/3$;
        \item If $\mathrm{dist}(U,\mathcal{V}_{k}) > 2\sqrt{2\epsilon}$, then $\mathcal{A}'$ outputs reject with probability at least $2/3$.
    \end{itemize}
    Furthermore, the algorithm $\mathcal{A}'$ uses $q(k,\epsilon,l)$ queries to $U$.

\end{proposition}

\begin{proof}
    The algorithm $\mathcal{A}'$ first generates a random partition $\mathcal{I}$ of $[n]$ into $l\overset{\mathrm{def}}{=} 24k^2$ parts by uniformly and independently assigning each element $i\in [n]$ to one of these parts. Subsequently, the algorithm $\mathcal{A}'$ calls upon $\mathcal{A}$ with parameters $\epsilon,k,l$ and the partition $\mathcal{I}$. According to \Cref{lemma:approximate be k-part} and \Cref{lemma:violate be k-part}, with probability at least $5/6$, the following statements are satisfied.
    \begin{enumerate} [(i)]
        \item If $\mathrm{dist}(U,\mathcal{V}_k) \le \sqrt{\epsilon'/2}$, then $U$ is $\epsilon'$-close to some quantum $k$-part junta with respect to $\mathcal{I}$;
        \item If $\mathrm{dist}(U,\mathcal{V}_{k}) > 2\sqrt{2\epsilon}$, then $U$ is $\epsilon$-far from every quantum $k$-part junta with respect to $\mathcal{I}$.
    \end{enumerate}
    Consequently, $\mathcal{A}$ will yield a response consistent with the proposition, with probability at least $5/6$, indicating that the success probability of $\mathcal{A}'$ is at least $2/3$, with $q(k,\epsilon,l)$ queries.
\end{proof}

\section{Tolerant Quantum Junta Tester } \label{section:solution}
This section presents a detailed description and analysis of our algorithm designed for quantum $k$-part junta testing. By \Cref{definition:partition quantum juntas}, it is sufficient to say that a unitary $U$ is $\epsilon$-close to some quantum $k$-part junta with respect to a partition $\mathcal{I}$ if we can find a set $S \subseteq [l]$ of size $l-k$ satisfying $\mathbf{Inf}_U[\phi_\mathcal{I}(S)] \le \epsilon$. It is costly to estimate the influences of all $\binom{l}{k} = 2^{(1+o(1))k\log l}$ subsets. We can reduce the query complexity if the same query to $U$ can be repeatedly used to estimate the influence of the set $\phi_\mathcal{I}(S)$ for different $S$. Therefore, we introduce the concept of random $\rho$-biased subset, and give the relation between the influence of  set $S$ and its $\rho$-subset influence in \Cref{subsection:random rho subset}. Then we provide an algorithm that estimates the $\rho$-subset influence of all $S\subseteq [l]$ of size $l-k$ in \Cref{subsection:estimate rho subset influence}, enabling us to determine if there exists a set $S\subseteq [l]$ of size $l-k$ satisfying $\mathbf{Inf}_U[\phi_\mathcal{I}(S)] \le \epsilon$. We accomplish our proof of \Cref{theorem:main result} in \Cref{subsection:proof of result}.

\subsection{Random \texorpdfstring{$\rho$}{}-biased Subset} \label{subsection:random rho subset}

\begin{definition}
   For $\rho \in (0,1)$ and a set $S$, we independently assign each element $i \in S$ to $S'$ with probability $\rho$, denoted by $S' \sim_\rho S$. Such a set $S'$ is called a random $\rho$-biased subset of $S$.
\end{definition}

\begin{definition}
    Given a partition $\mathcal{I} = \{I_1,\dots,I_l\}$ and a set $S\subseteq [l]$, the expected value of the influence of all random $\rho$-biased subsets of $S$, denoted by $\mathbf{E}_{S'\sim_\rho S}\left[\mathbf{Inf}_U[\phi_\mathcal{I}(S')]\right]$, is called the $\rho$-subset influence of $S$ with respect to $\mathcal{I}$.
\end{definition}

\begin{lemma} \label{lemma:rho influence and influence}
    Consider a partition $\mathcal{I} = \{I_1,\dots,I_l\}$ of $[n]$. For any set $S\subseteq [l]$, let $S'$ denote the random $\rho$-biased subset of $S$, then we have the following inequality:
    $$\frac{\rho}{3} \mathbf{Inf}_U[\phi_\mathcal{I}(S)] \le \mathbf{E}_{S'\sim_\rho S}\left[\mathbf{Inf}_U[\phi_{\mathcal{I}}(S')]\right] \le \mathbf{Inf}_U[\phi_\mathcal{I}(S)].$$
\end{lemma}

\Cref{lemma:rho influence and influence} describes a relation between the $\rho$-subset influence of $S$ and the influence of $S$, which is a quantum counterpart to Lemma 6.3 in \cite{blais2019tolerant}. Basically the proof of \Cref{lemma:rho influence and influence} is the same as \cite{blais2019tolerant}, for the proof only requires the two properties of influence in \Cref{lemma:properties of influence}. For completeness, we provide the proof of \Cref{lemma:rho influence and influence} in \Cref{section:proof of rho influence}.

\subsection{Estimate of The \texorpdfstring{$\rho$}{}-subset Influence} \label{subsection:estimate rho subset influence}

This section presents an algorithm to estimate the $\rho$-subset influence:  the  output $v_S^\rho$  of \Cref{algorithm:rho influence estimator} can approximate $\mathbf{E}_{S'\sim_\rho S}\left[\mathbf{Inf}_U[\phi_{\mathcal{I}}(S')]\right]$ to a sufficient accuracy for each $S\subseteq [l]$ with $|S|=l-k$.  More formally, we have \Cref{lemma:rho-bound}.

\begin{algorithm}[htb]
	\caption{$\rho$-subset Influence Estimator}
	\label{algorithm:rho influence estimator}
	\KwIn{Oracle access to $U$, partition $\mathcal{I}$ and parameters $\rho,\beta,\gamma,k,l$.}
	\KwOut{$v^\rho_S$ for each $S\subseteq [l]$ with $|S|=l-k$.}  
	\BlankLine

Set $m=\frac{C\cdot k \log l}{\gamma^2 \beta \rho (1-\rho)^k}$, where $C \ge 1$ is an absolute constant. 

\For{$i=1$ \KwTo $m$}{
Let $S^i\sim_\rho [l]$. 

Let $X^i$ represent the output of \Cref{algorithm:influence estimator} for $U$ and $\phi_{\mathcal{I}}(S^i)$.

}

Let $\mathcal{K}$ denote the multiset comprising subsets $S^1,\dots,S^m$. 

\For{each $S\subseteq [l]$ of size $l-k$}{
Let $\mathcal{K}_S \subseteq \mathcal{K}$ denote the collection of all sets $S^i \in \mathcal{K}$ such that $S^i \subseteq S$ .

Let $v^\rho_S \leftarrow \frac{1}{|\mathcal{K}_S|} \sum_{S^i \in \mathcal{K}_S} X^i$ serve as the estimate of the $\rho$-subset influence of $S$. 
}

\end{algorithm}

\begin{lemma} \label{lemma:rho-bound}
    Let $\mathcal{I}=\{I_1,\dots,I_l\}$ be a partition of $[n]$ and $S\subseteq [l]$ with $|S|=l-k$.  $S'$ is the random $\rho$-biased subset of $S$, that is, $S' \sim_\rho S$. For any $\beta,\rho \in (0,1)$ and a constant $\gamma \in (0,1)$, with probability at least $1-o(1)$, the  output $v_S^\rho$  of \Cref{algorithm:rho influence estimator} fulfill the following conditions:
    \begin{enumerate}
        \item if $\mathbf{E}_{S'\sim_\rho S}\left[\mathbf{Inf}_U[\phi_{\mathcal{I}}(S')]\right] > \frac{\rho \beta}{3}$, then 
        $$(1-\gamma) \cdot \mathbf{E}_{S'\sim_\rho S}\left[\mathbf{Inf}_U[\phi_{\mathcal{I}}(S')]\right] \le v^\rho_S \le (1+\gamma) \cdot \mathbf{E}_{S'\sim_\rho S}\left[\mathbf{Inf}_U[\phi_{\mathcal{I}}(S')]\right].$$
        \item if $\mathbf{E}_{S'\sim_\rho S}\left[\mathbf{Inf}_U[\phi_{\mathcal{I}}(S')]\right] \le \frac{\rho \beta}{4}$, then $$0 < v^\rho_S \le (1+\gamma) \frac{\rho \beta}{4}.$$
    \end{enumerate}
\end{lemma}

\begin{proof} The main idea of \Cref{algorithm:rho influence estimator} is as follows. Considering a partition $\mathcal{I} = \{I_1,\dots ,I_l\}$, we perform \Cref{algorithm:influence estimator} $m$ times for $m$ random $\rho$-biased subsets of $[l]$. Note that the expected value of the output by \Cref{algorithm:influence estimator} equals to the influence of the input set as \Cref{lemma:influence estimator} states. Then, With the $m$ outputs of \Cref{algorithm:influence estimator}, we can estimate the $\rho$-subset influences of all set $S \subseteq [l]$ of size $l-k$.

    Let $m' \overset{\mathrm{def}}{=} \frac{1}{2}(1-\rho)^k \cdot m = \frac{Ck\log l}{2\gamma^2 \beta \rho}$. We will first demonstrate that for any subset $S \subseteq [l]$ of size $l-k$, with probability at least $1-o(1)$, we have $|\mathcal{K}_S|\ge m'$. Considering a specific subset $S\subseteq [l]$ with $|S|=l-k$, for each index $i\in [m]$, let $\mathbbm{1}_{\{S^i \subseteq S\}}$ represent the indicator function, which takes the value $1$ if $S^i \subseteq S$, and $0$ otherwise. Then for each $i \in [m]$, we have $\Pr[\mathbbm{1}_{\{S^i \subseteq S\}}=1] = (1-\rho)^k$. By the Chernoff bound in \Cref{fact:chernoff bound}, we obtain the inequality:
    $$\Pr \left[ \sum_{i=1}^m \mathbbm{1}_{\{S^i \subseteq S\}} < \frac{1}{2} \cdot m(1-\rho)^k \right] \le e^{-\frac{m}{8}(1-\rho)^k} = e^{-\frac{C\cdot k\log l}{8\gamma^2 \beta \rho }} < 2^{-4k\log l},$$
    with a suitable choice of $C \ge 1$. By applying the union bound over all $\binom{l}{l-k} = \binom{l}{k}= 2^{(1+o(1))k\log l}$ sets, it holds that for every set $S\subseteq [l]$ of size $l-k$, we have $|\mathcal{K}_S| \ge m'$, with probability at least $1-o(1)$.

    By the definition of $v^\rho_S$ and \Cref{lemma:influence estimator}, we have
    \begin{equation} \label{eq:rho expect influence}
        \mathbf{E}[v^\rho_S] = \mathbf{E}_\mathcal{K}\left[\frac{1}{|\mathcal{K}_S|} \sum_{S^i \in \mathcal{K}_S} \mathbf{E}[X^i] \right] = \sum_{S'\subseteq S} \Pr[S'\in \mathcal{K}] \cdot \mathbf{Inf}_U[\phi_\mathcal{I}(S')] = \mathbf{E}_{S'\sim_\rho S}\left[\mathbf{Inf}_U[\phi_{\mathcal{I}}(S')]\right].
    \end{equation}
    
    Now considering a set $S$ with $\mathbf{E}_{S'\sim_\rho S}[\mathbf{Inf}_U[\phi_{\mathcal{I}}(S')]] > \frac{\rho \beta}{3}$, by a Chernoff bound and \Cref{eq:rho expect influence}, it can be concluded that
    \begin{equation*}
        \begin{aligned}
         \Pr\left[|v^\rho_S-\mathbf{E}_{S'\sim_\rho S}[\mathbf{Inf}_U[\phi_{\mathcal{I}}(S')]]| > \gamma \mathbf{E}_{S'\sim_\rho S}[\mathbf{Inf}_U[\phi_{\mathcal{I}}(S')]]\right] &\le 2e^{-\frac{|\mathcal{K}_S| \gamma^2 \cdot \mathbf{E}_{S'\sim_\rho S}[\mathbf{Inf}_U[\phi_{\mathcal{I}}(S')]]}{3}},  \\
         &\le 2e^{-\frac{m'\gamma^2 \rho \beta}{9}},\\
         &< 2^{-4k\log l},
        \end{aligned}
    \end{equation*}
    once more for an appropriate selection of $C \ge 1$. Taking the union bound over all sets $S\subseteq [l]$ of size $l-k$, consequently, for each $S$ where $\mathbf{E}_{S'\sim_\rho S}[\mathbf{Inf}_U[\phi_{\mathcal{I}}(S')]] > \frac{\rho \beta}{3}$, it follows that $v^\rho_S$ lies in the interval as indicated by the lemma, with probability at least $1-o(1)$.

    Let us consider a set $S \subseteq [l]$ with $|S|=l-k$ such that $\mathbf{E}_{S'\sim_\rho S}[\mathbf{Inf}_U[\phi_{\mathcal{I}}(S')]] \le \frac{\rho \beta}{4}$. Utilizing the Chernoff bound in conjunction with \Cref{eq:rho expect influence}, we have the following inequality:
    $$\Pr \left[v^\rho_S > (1+\gamma) \frac{\rho \beta}{4} \right] \le \Pr \left[ v^\rho_S > \mathbf{E}_{S'\sim_\rho S}[\mathbf{Inf}_U[\phi_{\mathcal{I}}(S')]] \right]  \le e^{-\frac{\gamma^2}{3}\frac{\rho \beta}{4}|\mathcal{K}_S|} \le e^{-\frac{\gamma^2 \rho \beta}{12}m'} < 2^{-4k \log l}.$$
    By applying the union bound across all sets $S\subseteq[l]$ with $|S|=l-k$, then for each $S$ satisfying the condition $\mathbf{E}_{S'\sim_\rho S}[\mathbf{Inf}_U[\phi_{\mathcal{I}}(S')]] \le \frac{\rho \beta}{4}$, we have $v^\rho_S \le (1+\gamma) \frac{\rho \beta}{4}$, with probability at least $1-o(1)$, as the lemma states.
\end{proof}

\begin{lemma}[\cite{chen2023testing}, Lemma 3.1] \label{lemma:influence estimator}
    Let $X$ represent the result produced by \Cref{algorithm:influence estimator} for $U\in \mathcal{U}_n$ and $S \subseteq [n]$, then we have
    $$\mathbf{E}[X] = \mathbf{Inf}_U[S].$$
\end{lemma}

\begin{algorithm}[htb]
	\caption{Influence Estimator for Quantum Unitaries}
	\label{algorithm:influence estimator}
	\KwIn{Oracle access to $U,\ S \subseteq [n]$.}
	\KwOut{$X\in \{0,1\}$.}  
	\BlankLine

\begin{enumerate}
    \item Prepare the CJ state
    $$|v(U)\rangle  =\frac{1}{\sqrt{2^n}} \sum_{0\le i,j < 2^n} U[i,j] |i\rangle |j\rangle,$$
    by querying $U$ once on the $n$ EPR pairs $(l,\tilde{l}) \in [n] \times \{n+1,\dots,2n\}$.
    \item Measure the $2|S|$ qubits in the set $S\cup \{\tilde{l}:l\in S\}$ in the Bell basis and denote the post-measurement state as $|\varphi \rangle$.
    \begin{enumerate}[(a)]
        \item If the state $|\varphi\rangle$ equals $|\mathrm{EPR}\rangle^{\otimes |S|}$, return $0$;
        \item Otherwise, return $1$.
    \end{enumerate}
\end{enumerate}
\end{algorithm}

\subsection{Proof of \texorpdfstring{\Cref{theorem:main result}}{}} \label{subsection:proof of result}

    According to \Cref{proposition:reduction to part juntas}, to obtain the result in \Cref{theorem:main result}, it is adequate to examine a $l= 24k^2$ partition $\mathcal{I}$, and to confirm that \Cref{algorithm:rho junta tester} can distinguish between the following two cases with probability at least $5/6$.
    \begin{enumerate}
        \item The unitary operator $U$ is $\frac{\epsilon^2 \rho}{32}$-close to some quantum $k$-part junta with respect to $\mathcal{I}$;
        \item The unitary operator $U$ is $\frac{\epsilon^2}{8}$-far from every quantum $k$-part junta with respect to $\mathcal{I}$.
    \end{enumerate}

    With $\beta = \frac{\epsilon^2}{8}$, we suppose that $U$ is $\frac{\rho \beta}{4}$-close to some quantum $k$-part junta with respect to $\mathcal{I}$, meaning that there exists a set $S\subseteq [l]$ of size $l-k$ where $\mathbf{Inf}_U[\phi_\mathcal{I}(S)] \le \frac{\rho \beta}{4}$. Furthermore, by applying \Cref{lemma:rho influence and influence}, we find that the $\rho$-subset influence $\mathbf{E}_{S'\sim_\rho S}\left[\mathbf{Inf}_U[\phi_\mathcal{I}(S')]\right] \le \frac{\rho \beta}{4}$, and by \Cref{lemma:rho-bound}, the estimate $v^\rho_S$ does not exceed $(1+1/8)\frac{\rho \beta }{4} = \frac{9\rho \beta}{32}$, with probability at least $1-o(1)$.

    Now we consider the case in which $U$ is $\beta$-far from every quantum $k$-part junta with respect to $\mathcal{I}$, meaning that for each set $S \subseteq [l]$ of size $l-k$, we have $\mathbf{Inf}_U[\phi_\mathcal{I}(S)] > \beta$. By applying \Cref{lemma:rho influence and influence}, we obtain the inequality $\mathbf{E}_{S'\sim_\rho S}\left[\mathbf{Inf}_U[\phi_\mathcal{I}(S')]\right] \ge \frac{\rho}{3} \cdot \mathbf{Inf}_U[\phi_\mathcal{I}(S)] > \frac{\rho \beta}{3}$. Consequently, by applying \Cref{lemma:rho-bound}, we find that for every set $S\subseteq[l]$ of size $l-k$, the following holds:
    $$v^\rho_S \ge \left(1-\frac{1}{8}\right) \mathbf{E}_{S'\sim_\rho S}\left[\mathbf{Inf}_U[\phi_\mathcal{I}(S')]\right] >  \frac{9\rho \beta}{32},$$
    with probability at least $1-o(1)$. The query complexity is the number of times we invoke \Cref{algorithm:influence estimator}, i.e., $m = O\left(\frac{k \log k}{\epsilon^2 \rho (1-\rho)^k}\right)$.

\begin{algorithm}[htb]
	\caption{Tolerant Quantum Junta Tester}
	\label{algorithm:rho junta tester}
	\KwIn{Oracle access to $U$, parameters $\rho,\epsilon,k$.}
	\KwOut{Accepts if $U$ is $\frac{\sqrt{\rho}}{8}\epsilon$-close to some quantum $k$-junta, rejects if $U$ is $\epsilon$-far from every quantum $k$-junta.}  
	\BlankLine

Construct a random partition $\mathcal{I}$ of the set $[n]$ consisting of $l=24k^2$ parts by uniformly and independently including each qubit $i \in [n]$ in one of these parts.

Invoke \Cref{algorithm:rho influence estimator} utilizing the partition $\mathcal{I}$, with parameters set as $\rho = \rho$, $\beta = \epsilon^2 / 16$, $\gamma = 1/8$, $k = k$, $l=24k^2$.

\For{each $S\subseteq [l]$ of size $l-k$}{
    \If{$v^\rho_S \le \frac{9\rho \beta}{32}$}{
        \Return accept.
    }
}

\Return reject.

\end{algorithm}

 \bibliographystyle{alpha} 
\bibliography{reference}

\appendix
\section{Proof of \texorpdfstring{\Cref{lemma:rho influence and influence}}{}} \label{section:proof of rho influence}

\begin{definition} [Legal collection of covers]
    Let $S$ be a set comprising $s$ elements, and for any $s'\in [s]$, consider the family $\mathcal{F}_S^{s'}$ consists of all subsets of $S$ with size $s'$. We say that $\mathcal{C} \subseteq \mathcal{F}_S^{s'}$ is a cover of $S$ if $\bigcup_{Y \in \mathcal{C}} Y = S$. Furthermore, we say that a collection of covers $\mathscr{C}_S = \{\mathcal{C}_1,\dots, \mathcal{C}_r\}$ is a legal collection of covers for $S$ if each $\mathcal{C} \in \mathscr{C}_S$ is a cover of $S$ and these covers are mutually disjoint.
\end{definition}

We are interested in identifying a legal collection of covers for $S$ whose size is “as big as possible”. \Cref{lemma:legal collection} guarantees that there exists such a cover achieving the optimal size.

\begin{lemma} [\cite{baranyai1974factrization,blais2019tolerant}] \label{lemma:legal collection}
    For any set $S$ of $s$ elements, there exists a legal collection of covers $\mathscr{C}_S$ of size at least
    $$|\mathscr{C}_S| \ge \left\lfloor \frac{\binom{s}{s'}}{\left\lceil \frac{s}{s'} \right\rceil} \right\rfloor.$$
    Moreover, this is a tight bound.
\end{lemma}

Observe that 
$$\left\lfloor \frac{\binom{s}{s'}}{\left\lceil \frac{s}{s'} \right\rceil} \right\rfloor = \left\lfloor \frac{\frac{s}{s'}}{\left\lceil \frac{s}{s'} \right\rceil}\binom{s-1}{s'-1} \right\rfloor \ge \left\lfloor \frac{\frac{s}{s'}}{\frac{s}{s'}+1}\binom{s-1}{s'-1} \right\rfloor \ge \left\lfloor \frac{1}{2} \binom{s-1}{s'-1} \right\rfloor \ge \frac{1}{3} \binom{s-1}{s'-1}.$$
Let $s=|S|$, then we have
$$\mathbf{E}_{S' \sim_\rho S} \left[\mathbf{Inf}_U[\phi_\mathcal{I}(S')]\right] = \sum_{s'=1}^s \sum_{S'\subseteq S:|S'|=s'} \rho^{s'} (1-\rho)^{s-s'} \cdot \mathbf{Inf}_U[\phi_\mathcal{I}(S')].$$
By \Cref{lemma:legal collection},
\begin{equation*}
    \begin{aligned}
\sum_{S'\subseteq S:|S'|=s'}\mathbf{Inf}_U[\phi_\mathcal{I}(S')] &= \sum_{S'\in \mathcal{F}_S^{s'}} \mathbf{Inf}_U[\phi_\mathcal{I}(S')] \ge \sum_{\mathcal{C} \in \mathscr{C}_S} \sum_{S'\in \mathcal{C}} \mathbf{Inf}_U[\phi_\mathcal{I}(S')] \\
&\ge |\mathscr{C}_S| \cdot \mathbf{Inf}_U[\phi_\mathcal{I}(S)] \ge \frac{1}{3}\binom{s-1}{s'-1} \mathbf{Inf}_U[\phi_\mathcal{I}(S)].
    \end{aligned}
\end{equation*}
Consequently, we derive the lower bound
\begin{equation*}
    \begin{aligned}
        \mathbf{E}_{S' \sim_\rho S} \left[\mathbf{Inf}_U[\phi_\mathcal{I}(S')]\right] &\ge \sum_{s'=1}^s \rho^{s'} (1-\rho)^{s-s'} \cdot \left(\frac{1}{3}\binom{s-1}{s'-1} \mathbf{Inf}_U[\phi_\mathcal{I}(S)]\right) \\
        &=\frac{\rho}{3} \mathbf{Inf}_U[\phi_\mathcal{I}(S)] \cdot \sum_{s'=1}^s \binom{s-1}{s'-1} \rho^{s'-1} (1-\rho)^{s-s'}\\
        &=\frac{\rho}{3} \mathbf{Inf}_U[\phi_\mathcal{I}(S)] \cdot [\rho + (1-\rho)]^{s-1}\\
        &=\frac{\rho}{3} \mathbf{Inf}_U[\phi_\mathcal{I}(S)].
    \end{aligned}
\end{equation*}

As for the upper bound, due to the monotonicity of influence, it follows that $\mathbf{Inf}_U[\phi_\mathcal{I}(S')] \le \mathbf{Inf}_U[\phi_\mathcal{I}(S)]$ for any $S' \subseteq S$. Then the upper bound $\mathbf{E}_{S' \sim_\rho S} \left[\mathbf{Inf}_U[\phi_\mathcal{I}(S')]\right] \le \mathbf{Inf}_U[\phi_\mathcal{I}(S)]$ can be easily derived.

\end{document}